\def\draft{0} %
\def\llncs{0}
\def\anon{0}

\ifnum\llncs=0
\documentclass[11pt]{article}
\usepackage{fullpage}
\else
\documentclass[orivec,runningheads,envcountsect,envcountsame]{llncs}
\fi

\usepackage{multirow}

\ifnum\draft=1
\def\ShowAuthNotes{1}
\else
\def\ShowAuthNotes{0}
\fi

\newcommand{\remove}[1]{}
\newcommand{\ignore}[1]{}
\ifnum\llncs=0
\usepackage{amsthm}
\fi

\usepackage{amsmath,amssymb}
\usepackage{url}
\usepackage[sectionbib,numbers]{natbib}
\usepackage{multicol}
\usepackage{bbm}

\usepackage{color}
\definecolor{DarkBlue}{RGB}{0,0,150}
\definecolor{DarkRed}{RGB}{150,0,0}
\definecolor{darkgreen}{rgb}{0.0, 0.2, 0.13}
\usepackage[colorlinks,linkcolor=DarkBlue,citecolor=DarkBlue]{hyperref}

\ifnum\llncs=0

\newtheorem{theorem}{Theorem}[section]
\newtheorem{proposition}[theorem]{Proposition}
\newtheorem{definition}[theorem]{Definition}
\newtheorem{claim}[theorem]{Claim}
\newtheorem{lemma}[theorem]{Lemma}

\newtheorem{remark}[theorem]{Remark}
\newtheorem{corollary}[theorem]{Corollary}

\newtheorem{test}{Test}[section]

\usepackage{tikz}
\else

\fi

\usepackage{tikz} 
\usetikzlibrary{automata,arrows,positioning,calc}
\usepackage{graphicx}

\def \w {\omega}

\def\Z{\mathbb{Z}}

\def \q {\mathfrak{q}}

\def \lra {\longrightarrow}

\def\cC{{\cal C}}

\def\cG{{\cal G}}

\def\cI{{\cal I}}

\def\cP{{\cal P}}

\def\cS{{\cal S}}

\def\bbC{{\mathbb C}}
\def\bbE{{\mathbb E}}
\def\bbF{{\mathbb F}}

\def\bbN{{\mathbb N}}

\def\bbZ{{\mathbb Z}}
\def\Z{{\mathbb Z}}

\def\binset{\{0,1\}}

\def\q2{\lfloor q/2 \rceil}

\newcommand{\abs}[1]{\left\vert {#1} \right\vert}
\newcommand{\norm}[1]{\left\| {#1} \right\|}

\def\poly{{\rm poly}}

\ifnum\ShowAuthNotes=1
\newcommand{\authnote}[3]{\textcolor{#3}{[{\footnotesize {\bf #1:} { {#2}}}]}}
\else
\newcommand{\authnote}[3]{}
\fi

\newcommand{\ket}[1]{|{#1}\rangle}
\newcommand{\bra}[1]{\langle{#1}|}

\newcommand{\tr}{\tau}
\newcommand{\ntr}{\hat{\tau}}

\def \ra {\rightarrow}

\def \C {\mathbb{C}}

\def \Z {\mathbb{Z}}

\newcommand{\qedllncs}{\ifnum\llncs=1{$\square$}\fi}

\title{Unitary Subgroup Testing}

\ifnum\llncs=1

\ifnum\anon=1
\author{}
\institute{}
\else

\author{Zvika Brakerski
	\and Devika Sharma
	\and Guy Weissenberg}
\institute{}

\fi

\renewcommand{\paragraph}{\boldpar}

\else
\ifnum\anon=1
\author{}
\else

\author{
	Zvika Brakerski \thanks{Weizmann Institute of Science, Israel,  \texttt{\{zvika.brakerski,devika.sharma,guy.weissenberg\}@weizmann.ac.il}.	
		Supported by the Binational Science Foundation (Grant No.\ 2016726), and by the European Union Horizon 2020 Research and Innovation Program via ERC Project REACT (Grant 756482) and via Project PROMETHEUS (Grant 780701).}
	\and Devika Sharma \footnotemark[1]
	\and Guy Weissenberg \footnotemark[1]
}
\fi
\ifnum\draft=1
\date{\today}
\else
\date{}
\fi

\fi

\begin{document}

\maketitle
\begin{abstract}
We consider the problem of \emph{subgroup testing} for a quantum circuit $C$: given access to $C$, determine whether it implements a unitary that is $a$-close or $b$-far from a subgroup $\cG$ of the unitary group. It encompasses the problem of exact testing, property testing and tolerant testing. In this work, we study these problems with the group $\cG$ as the trivial subgroup (i.e.\ identity testing) or the Pauli or Clifford group and their $q$-ary extension, and a \emph{promise} version of these problems where $C$ is promised to be in some subgroup of the unitaries that contains $\cG$ (e.g.\ identity testing for Clifford circuits).

Our main result is an equivalence between Pauli testing, Clifford testing and Identity testing. We derive the equivalence between Clifford and Identity testing by showing a structural property of the Clifford unitaries. Namely, that their (normalized) trace lies in the discrete set $\{2^{-k/2}: k \in \bbN\} \cup \{0\}$, regardless of the dimension. We also state and prove the analogous property for the $q$-ary Cliffords. This result allows us to analyze a very simple single-query identity test under the Clifford/Pauli promise. To prove the equivalence between Pauli and Identity testing, we analyze the conjugation action of a non-Pauli unitary on the Pauli group and show that its distance from the Pauli group affects the number of fixed points. We believe that these results are of interest, independent of their application to derive the equivalence considered in this work.

We use the equivalence to compare (and thus establish) computational hardness for the problems of Pauli and Clifford testing.

\end{abstract}

\section{Introduction}
In the current NISQ-era, closely monitoring the evolution of a quantum system is an essential task, and that makes property testing elemental. Labs maintaining and building quantum infrastructure, in practice, use the standard approach of quantum process tomography (QPT) \cite{nielsen-chung} to learn full information about the quantum process. A 
similar problem,
that of testing whether a unitary operator satisfies a certain property, is considered in this work. Given a subgroup $\cG$ of the unitary group and a unitary operator $C$, we consider the question of deciding how far or close $C$ is from $\cG$, w.r.t a specified distance measure. For the choice of the group, we restrict to $\cG$ being (the $q$-ary versions of) the Identity, the Pauli or the Clifford group, and given the central role the groups play in error correction and fault tolerance computing, we believe it is a good choice for a set of subgroups.

We consider three avatars of the subgroup testing problem. First, the exact version, where one decides whether $C$ belongs to $\cG$. This version is independent of any distance measure. Second, the standard setup of property testing, where one decides if $C$ is in $\cG$ or is at least $b$-far away from all elements of $\cG$, w.r.t a specified distance measure. See \cite{wangPropertyTesting,montanaro2018survey}. Finally, we consider the most general version, known as tolerant testing: the problem of deciding if $C$ is $a$-close ($a>0$) to an element of $\cG$ or is $b$-far away from all elements of $\cG$.

Our subgroup testing problem generalizes the well-researched identity-testing problem, wherein one is 
required to decide whether the unitary operator $C$ is equivalent to the identity operator. The problem of
exact identity testing was proven to be co-NQP-complete \cite{tanaka2009exact}, whereas the problem of tolerant testing for the trivial\footnote{Throughout this work, we consider unitaries with identical functionality to be equivalent, namely we ignore global phase. Thus, the ``trivial group'' we refer to, in fact, contains all unitaries that are equal to the identity up to a  global phase. The same holds for all other subgroups that we discuss.} subgroup
in the operator norm,  
was proven to be co-QMA-complete in~\cite{NonIDtesting}.
It is, therefore, natural to inquire about the hardness of subgroup testing beyond the trivial subgroup.

We present a new angle to investigate the hardness of subgroup testing via means of reduction.
We reduce the problem of 
subgroup testing, $\cG$-testing, for a subgroup $\cG$, to two independent problems: the problem of $\cS$-testing, for a subgroup $\cS \supseteq \cG$, and a promise $(\cS, \cG)$-testing, a \emph{promise} version of the $\cG$-testing, where $C$ is guaranteed to not be an arbitrary unitary operator but rather to belong to the bigger subgroup $\cS$. This is a novel contribution of our work. 

We show that identity-testing (i.e.,\ when $\cG$ is the trivial group) is efficiently decidable when the promise group $\cS$ is the Pauli or the Clifford group. Due to its simplicity and efficiency, this test can be used as sanity checks while implementing an error correction code using Pauli and Clifford gates.  We devise separate (promise-) identity testers for Pauli/Clifford based on how one may access the input unitary operator $C$. If we are given a description of $C$ as a circuit containing gates from the Clifford group, we use the explicit description of the action of $C$ on a formal Pauli generator to determine if $C$ is equivalent to the identity. See Test~\ref{test:WB} for details. This test achieves perfect completeness and soundness (Theorem~\ref{thm:WBTest}), but we view this requirement of the input $C$ as quite restrictive. If, however, we are only allowed query access to the unitary $C$, we describe yet another test (Black-box Test~\ref{test:eprtest}) that requires a single query to $C$, yields perfect completeness and soundness that is simply the absolute value squared of the (normalized) trace of the unitary $C$. See Theorem~\ref{thm:BBTest} for details. This test does appear (in some form) in prior literature, in the binary case. However, our tests are generalized to hold true for qudits and the corresponding $q$-ary subgroups, for an odd prime $q$. The novelty of our work lies in the analysis of the soundness of our Black-box test, that exposes a property of the Clifford unitaries that we did not find in prior literature. Namely, that the trace of a Clifford unitary, in absolute value, must either be $0$ or a power of $\sqrt{q}$ (or $\sqrt{2}$ for binary Cliffords). We view this as a significant contribution of this work and hope that it will find other applications. The proof is based on straight-forward, yet clever modular arithmetic. We stress that our tests hold also for the qudit versions of Pauli and Clifford operators~\cite{Clark_2006,Farinholt} (for prime arity qudits), which are quite different in terms of functionality from their binary counterparts.  

As mentioned before, we observe that the promise $(\cS, \cG)$-subgroup testing can be viewed as a reduction from $\cG$-testing to $\cS$-testing. We give a general recipe, a.k.a the composition algorithm (Definition~\ref{algo:composition}), for constructing an algorithm for the $\cG$-testing problem by composing the algorithms for $\cS$-testing and the promise $(\cS, \cG)$-testing problems. In Section~\ref{subsec:ITP-PTP/CTP}, we use our identity-testers for Pauli/Clifford to show that the Identity-testing problem reduces to both the Pauli-testing and Clifford-testing problems, in all three avatars and w.r.t two distinct distance measures -- one induced by the operator norm, denoted $D^{op}$, and the second `average-case' distance norm, denoted $D^{tr}$, defined using the normalized trace of the unitary operators involved. See Section~\ref{sec:prelim} for the definition of the measures. 

Further, we give additional reductions from the problems of Clifford testing to the Pauli testing and the problems of Pauli testing to Identity testing in the setup of property testing and w.r.t the distance measure $D^{tr}$. 
To reduce the problem of Pauli testing to Identity testing, we analyze the conjugation action of a unitary on the group of unitaries $U_{q^n}(\C)$ to prove that a non-Pauli unitary can have at most $\frac{1}{q} \cdot |\cP_{q}^n|$ fixed points in the Pauli group (denoted $\cP_{q}^n$). See Lemma~\ref{lem:pauli_commute}. In fact, we show that the distance of $U$ from Pauli, $D^{tr}(C, \cP)$ determines the distance between the image and the pre-image, $D^{tr}(CPC^{\dagger},P)$, for a fraction of $P \in \cP_q^n$. See Lemma~\ref{lem:approximate_commutator}. We believe that these results maybe of general interest, independent of the context in this work.

As a consequence, the reductions imply that the exact versions of the Pauli and Clifford testing problems are at least as hard as the identity testing problem and therefore are co-NQP-hard, as exact Identity testing is co-NQP-complete \cite{tanaka2009exact}. Further, these exact problems are in fact equivalent under randomized reductions. Secondly, in the distance measure $D^{op}$, the problems of Pauli and Clifford testing (in the setup of tolerant testing) are at least as hard as the Identity testing problem and hence are both QMA hard (under Turing reductions).  See \cite{NonIDtesting}. Finally, in the distance measure $D^{tr}$, the three problems are equivalent under randomized reductions.

\paragraph{Comparison to previous works} In \cite{wangPropertyTesting}, the author gives algorithms for the Identity, Pauli and Clifford testing problems in the setup of property testing over qubits. The distance measure used by the author is equivalent to the distance measure $D^{tr}$ that we consider. Our reduction from Clifford testing to Pauli testing is an adaptation of \cite[Algorithm 4]{wangPropertyTesting} to the $q$-ary setting. 

In a concurrent work, \cite{Linden2021lightweight}, the authors discuss similar identity testing algorithms for Clifford circuits acting on qubits (i.e., for the task of deciding whether two Clifford circuits are identical). 
The correctness of their tests rely on results (\cite[Theorem 4 $\&$ 5]{Linden2021lightweight}) that are both special cases of our commutator lemma (Lemma \ref{lem:pauli_commute}) and our result about the discreteness of Clifford traces (Lemma \ref{thm:cliffordtraceoddq}), respectively. In particular, Theorem $5$ in~\cite{Linden2021lightweight} gives an upper bound on the trace of non-identity Clifford elements acting on qubits, while we show that these traces, in fact, form a discrete set. The discreetness of Clifford traces (in the binary) was also proved in \cite{Bravyi_2016}. Their proof is conceptually different from ours.

\section{Preliminaries}\label{sec:prelim}

\paragraph{Complexity theory} For (possibly randomized) algorithms for (possibly promise) decision problems, we use standard terminology. This text is fairly self-contained. For a detailed description, we recommend, \cite[Section 2, 6]{GoldreichComplexity}. As usual, completeness is the minimal probability of accepting a YES instance and soundness is the maximal probability of accepting a NO instance. Perfect completeness is completeness, $1$ and perfect soundness is soundness $0$. The difference between completeness and soundness is called the soundness gap of the algorithm. If the soundness gap is at least $\epsilon$ then it can be amplified, via standard repetition, to $1-\delta$ using $\poly(1/\epsilon, \log(1/\delta))$ many calls to the algorithm.
\paragraph{Qudits and the Pauli group} We also use standard quantum computing notation in our text. We briefly define the notations and definitions we use. 
Qudits are elements in a $q^n$-dimensional (complex) Hilbert space generated by the computational basis $\{\ket{i_1\ldots i_n} : i_j \in [q] \}$ 
over $\C$. Norm-one elements in this space, i.e., 
\[\ket{\psi} = \sum_{i\in B} a_i \ket{i}, \quad \quad \hbox{ with } \quad \sum_{i \in B} |a_i|^2=1\]
are called pure qudits. Operators on qudits are elements in $U_{q^n}(\C)$, i.e., $q^n \times q^n$ unitary matrices with entries in $\C$. Given a unitary operator $U$, we denote its inverse by $U^\dagger$. Let $\cI := \{e^{i \varphi} I\}_{\varphi \in [-\pi, \pi)}$ denote the subgroup of unitaries that are functionally equivalent to identity (where the dimension is clear from context). For a subgroup $\cG$ of unitaries, we let $\cI\cG$ be the product of the two subgroups. Among the unitary operators, the Pauli matrices define the first (and sufficient) building blocks for elements in $U_q(\C)$. 
We say sufficient, because any  
operator can be generated as a complex linear combination of the Pauli matrices (of the same rank). See Theorem~\ref{thm:Pauli_decomp}. 
Recall that the binary Pauli group is the set $\mathcal{P}_2^1= \{\w_4^cX^aZ^b | \ a, b \in \Z_2, \ c \in \Z_4\}$, where $\w_4$ is the primitive $4$-th root of unity, $X \ket{x}:=\ket{x+1\mod 2}$, $Z \ket{x}:= (-1)^{x}\ket{x}$, and $Y := iXZ$. These operators generalize to the $q$-ary set up. 
For an odd prime $q$, the $q$-ary Pauli group, acting on a system of one qudit, is the set 
\[
\mathcal{P}_q^1=\{\omega_q^c X_q^aZ_q^b\ | \ a,b,c\in \bbZ_q\}~.
\]
where, $X_q: \ket{x}\ra \ket{x+1 \mod q} \quad \hbox{and} \quad Z_q: \ket{x} \ra \w_q^x\ket{x}$,
with $\w_q$ is a primitive $q$-th root of unity.
The generators $X_q$ and $Z_q$ are of order $q$, i.e., $X_q^q = Z_q^q =I$ and satisfy $X_qZ_q = \w_q^{-1}Z_q X_q$. 
In a system of $n$ qudits, the computational basis of the $q^n$-dimensional complex Hilbert space is the set $\{\ket{i_1i_2\cdots i_q}: i_j \in \Z_q, \hbox{ for } 1\leq j\leq q \}$ and the group $\mathcal{P}_q^n$ of 
$q$-ary Pauli operators acting on $n$ qudits is defined as
\begin{eqnarray*}
\mathcal{P}^n_q &=& \{\w_q^r P_1\otimes ...\otimes P_n | \; r \in \Z_q, \ P_i \in \mathcal{P}_q^1, \; \hbox{ for } i \in [n]\}, \quad \hbox{ for an odd prime $q$ and}\\
\mathcal{P}^n_2 &=& \{\w_4^r P_1\otimes ...\otimes P_n | \; r \in \Z_4, \ P_i \in \mathcal{P}_2^1, \; \hbox{ for } i \in [n]\}, \quad \hbox{ when $q=2$.}
\end{eqnarray*}

\paragraph{Pauli decomposition} 
Let $M_{q^n}(\C)$ denote the $\C$-vector space of $q^n \times q^n$ square matrices with complex entries. We show that a subset of $\mathcal{P}_q^n$ forms a basis for $M_{q^n}(\C)$. Since the set of unitary matrices is a subset of $M_{q^n}(\C)$, it is in this sense that we called the Pauli group, a sufficient building block for the set of operators on qudits.
\begin{theorem}[folklore]\label{thm:Pauli_decomp}
The set $\mathcal{B} = \mathcal{B}_q^n:= \{ P_1 \otimes P_2 \otimes \cdots \otimes P_n: \; P_i = X_q^{j_i}Z_q^{k_i}, \; 0 \leq j_i, k_i \leq q-1, \; \forall i\}$ forms an orthonormal basis of $M_{q^n}(\C)$ w.r.t the 
inner product $\langle A, B \rangle := \ntr(A^\dagger B)$. Therefore, any $M \in M_{q^n}(\C)$ decomposes uniquely as $M = \sum_{i_j} m_{i_1, \cdots i_n} P_{i_1} \otimes \cdots \otimes P_{i_n}$, where $m_{i_1, \cdots i_n} = \langle \otimes_j P_{i_j}, M \rangle$.
\end{theorem}
\begin{proof}
Note that the $q^{2n}$ Pauli elements of the set $\mathcal{B}$ are $\C$-linearly independent. Since the cardinality of $\mathcal{B}$ equals the dimension of $M_{q^n}(\C)$, as a $\C$-vector space, the set $\mathcal{B}$ forms a basis of $M_{q^n}(\C)$. 

To show orthonormality of the elements of $\mathcal{B}$, we first show it for $n=1$. Note that $\ntr(X_q^aZ_q^b)=0$, if $a \neq 0$, as $X_q$ is a permutation matrix. Further, if $a = 0$ and $b \neq 0$, then the diagonal entries sum to $\sum_{i=0}^{q-1} w_q^i =0$. Therefore, unless $a=0=b$, i.e., unless $X_q^aZ_q^b = I$, the trace is always $0$. This implies that 
\[
\langle A, B \rangle = \ntr(A^\dagger B) =\begin{cases}
1 & A=B\\
0 &\text{otherwise}
\end{cases}
\]
In order to prove orthonormality for the general $n$, let $\otimes_{i=1}^nP_i$ and $\otimes_{i=1}^nQ_i$ be two distinct elements in $\mathcal{B}$. Using the following two well-known properties of the tensor product: $(A \otimes B)^\dagger = A^\dagger \otimes B^\dagger$, and $\tr(A\otimes B)= \tr(A) \cdot \tr(B)$, we conclude that 
\[\langle \otimes_i P_i, \; \otimes_i Q_i \rangle  =  \ntr((\otimes_i P_i)^\dagger (\otimes_i Q_i)) =\prod_i \ntr(P_i^\dagger Q_i)\]
which is non-zero, and equals $1$, if and only if $P_i = Q_i$, for all $i$. The expression for $M \in M_{q^n}(\C)$ and its uniqueness follows. 
\end{proof}

\paragraph{The $q$-ary Clifford group}
Similar to case of $q=2$, the $q$-ary Clifford Group is the normalizer of $\mathcal{P}_q^n$ in $U_{q^n}(\C)$, i.e., 
\[
\mathcal{C}^n_q=\left\{V\in U_{q^n} \ | \  V P V^{\dagger}\in \mathcal{P}^n_q, \ \ \forall P\in \mathcal{P}^n_q \right\}
\]
In the binary case, the Clifford group is generated by the unitary gates, $H \ket{x} = \frac{1}{\sqrt{2}}\sum_a \w_2^{ax}\ket{a}$, $CNOT\ket{xy} = \ket{x \ x+ y \mod 2}$ and the phase gate $S\ket{x} = \w_4^x \ket{x}$. 
Generalization of these gates to the $q$-ary setup, for an odd prime $q$, as $F_q$ (generalization of $H$), $CNOT_q$ and $S_q$ generate the $q$-ary Clifford group $\mathcal{C}_q^n$. Notice that while $F_q$ and $S_q$ are $1$-qudit gates and $CNOT_q$ is a $2$-qudit gate, one can tensor them with the identity operator in $n-1$ ($n-2$, resp.) remaining places so they lie in $\mathcal{C}_q^n$. It is these elements that form the generating set.
\begin{theorem}[{\cite[Cor.\ 7.12]{Clark_2006}}]\label{thm:clifgen} For an odd prime $q$, the $q$-ary Clifford group, $\mathcal{C}_q^n$, is generated by $\{ F_q, CNOT_q, S_q\}$, where $F_q\ket{x}= \frac{1}{\sqrt{q}}\sum_{a \in \Z_q}\w_q^{ax}\ket{a}$, $CNOT_q \ket{x,y}= \ket{x,x+y\mod q}$, and $S_q \ket{x}= \omega_q^{x(x-1)/2}\ket{x}$. 
\end{theorem}

The following is an immediate consequence of the definition of the set of generators.
\begin{theorem}[{\cite[Section 4.3]{Farinholt}}]\label{thm:clifconjugate}
Let $C$ be a quantum circuit consisting of $\{ F_q, CNOT_q, S_q\}$ gates and acting on $n$ $q$-ary qudits. Then, given the circuit $C$, there exist \emph{efficiently computable} polynomials $f,g,h$
s.t.\ for all $a,b \in \bbZ_q^n$ it holds that
\begin{align*}
    C X^{a}Z^{b} C^{\dagger} = \omega_{q'}^{h(a,b)} X^{f(a,b)}Z^{g(a,b)}~,
\end{align*}
where $q'=q$ if $q$ is an odd prime, and $q'=4$ if $q=2$.
Furthermore, $f,g$ are linear (have total degree at most $1$) and $h$ is quadratic (has total degree at most $2$). 
\end{theorem}

\paragraph{Distance measures} In order to check for closeness between unitary operators, one needs a distance measure. Here, we define the ones we use. For a detailed discussion on distance measures, see \cite[Section 5.1.1]{montanaro2018survey}. 

For any $d\times d$ square matrix $A$, we denote its trace by $\tr(A)$ and its \textit{normalized} trace by $\ntr(A):= \tr(A)/d$.  For convenience, in the rest of the text, if $M$ is a matrix and $v$ is an $n$-dimensional vector, we use the shorthand $M^v$ to denote $\otimes_{i=1}^n M^{v_i}$. To measure the distance between two unitary operators, $U$, $V$, we consider the following two distance measures.
\begin{itemize}
    \item $D^{op}(U,V):= \norm{U-V}_{op}$, where $\norm{\cdot}_{op}$ denotes the operator norm. Recall that, for a matrix $A$, the operator norm is defined as $\norm{A}_{op} = \inf\{c\ge 0 \text{ : $\norm{Av}\le c \norm{v}$, for all $v\in \bbC^{d}$}\}$.
    \item $D^{tr}(U,V) = \sqrt{1-\abs{\ntr(UV^{\dagger})}^2}$, where $\ntr(U)$ denotes the normalized trace of $U$.
\end{itemize}

Finally, we state the following useful fact - 
the operator norm of a matrix upper bounds the normalized trace. This is a special case of Hölder's inequality.

\begin{lemma} \label{lem:matrixnorms}
For a square matrix $A$, $\abs{\ntr(A)}\le \norm{A}_{op}$.
\end{lemma}

\subsection{The EPR Identity Test}
\label{sec:prelim:epridtest}

We consider the following test for identity. This test is derived from the Choi-Jamio\l{}kowski isomorphism as was previously done in \cite[Section C-1]{wangPropertyTesting}. In this paper,  
we refer to this test as the ``EPR-Identity-Testing''. We show in this work that this test has good soundness over certain subgroups of unitaries.

\begin{test}[EPR Identity-Test]\label{test:eprtest} Given query access to a black-box implementing a quantum unitary operator $U$ over $n$ qudits, prepare the state $\ket{e}:=\frac{1}{q^{n/2}} \sum_{x \in \bbF_q^n} \ket{x}\ket{x}$ over two $n$-dimensional registers $A$ and $B$. Call $U$ on the $A$ register and perform the projective measurement $\ket{e}\bra{e}$ on the registers $AB$.
\end{test}

The test is efficiently implementable; generating the state $\ket{e}$ requires a linear number of $q$-ary Clifford operators since $\ket{e} = (CNOT_q (F_q \otimes I) \ket{00}_{AB})^{\otimes n}$. Projecting onto $\ket{e}$ is therefore also efficient.
The test makes a single query to the unitary $U$. The following lemma relates the acceptance probability of the EPR test on $U$ to the trace of $U$. 
\begin{lemma}\label{lem:eprtest}
The acceptance probability of the EPR-Identity-Test on a unitary $U$ is $\abs{\ntr(U)}^2$.
\end{lemma}
\begin{proof} Calling $U$ on $A$ is equivalent to applying $U\otimes I$ on $\ket{e}$. Therefore, by definition, the probability that the test accepts is $\abs{\bra{e} (U \otimes I)\ket{e}}^2$. It follows from a simple computation that this expression equals $\abs{\ntr(U)}^2$. 
\end{proof}

\section{The commutator of Pauli and non-Pauli matrices}\label{sec:commute} 
A unitary operator $C$ acts on the group of unitaries $U_{q^n}(\C)$, via conjugation. In this section, we restrict the subgroup to be acted on, to the Pauli group and discuss the behaviour of the image set under the conjugation action by $C$. 

Recall that it is the Pauli group that fixes itself under the conjugation action. That is to say that the Pauli group is the biggest subgroup of the unitaries to commute with the Paulis. In fact, if $C$ is a non-Pauli unitary, then at most $\frac{1}{q}$-th of the Paulis maybe fixed under the conjugation action of $C$, as we show in Lemma~\ref{lem:pauli_commute} below. In Lemma~\ref{lem:approximate_commutator}, we analyze the non-fixed points under this action. 
We show that the distance of $C$ from the Pauli group $\cP_q^n$, in the measure $D^{tr}$, plays a role in determining the distance of the image $CPC^{\dagger}$ from $P$, for a fraction of the Paulis $P \in \cP_q^n$. We believe that these simple and insightful results maybe of interest, independent, to this work as well. 
\begin{lemma}\label{lem:pauli_commute}
Let $C\in GL_{q^n}(\C)\setminus \mathcal{P}_q^n$. Then, 
\[Pr_{P\in \cP_q^n}\left([C,P] :=C^{\dagger}P^{\dagger}CP \in \cI\right) \le \frac{1}{q}.
\]
\end{lemma}

\begin{proof}
Let $P=X^x Z^z$ be a Pauli element in $\mathcal{B}$, and let $C= \sum_{a,b}\alpha_{{a},{b}} \; X^{a} Z^{b}$ be the Pauli decomposition of $C$. See Theorem~\ref{thm:Pauli_decomp}. Recall that $x,z,a,b$ are all vectors in $\bbZ_q^n$. Then, the commutation relations of the Pauli generators imply that
\begin{align*}
    P^\dagger CP = \sum_{{a},{b}}\alpha_{{a},{b}}\omega_q^{\langle b, x \rangle - \langle a, z \rangle}X^{a}Z^{b}~.
\end{align*}
It holds that $[C,P] \in \cI$ if and only if $P^\dagger CP = C$, up to a global phase. Since the Pauli decomposition is unique, this means that for all $a,b$ for which $\alpha_{a,b} \neq 0$ it holds that $\langle b, x \rangle - \langle a, z \rangle = \gamma \pmod{q}$ for some global constant $\gamma$. Note that this is only possible if $\gamma \in \bbZ_q$. Namely, we get a homogeneous linear equation in the $2n+1$ variables $x,z,\gamma$ as $\langle (x, z, \gamma), (b,-a,-1) \rangle=0\pmod{q}$.

Since $C$ is not Pauli, it has at least two nonzero $\alpha_{a,b}$ values. Due to the $(-1)$ entry, the vectors $\{(-a,b,-1)\}_{a,b \in \bbZ_q^n}$ are pairwise linearly independent. Therefore, the rank of this system of equations is at least $2$ and the probability that a random assignment satisfies it is at most $q^{-2}$. Taking a union bound over all possible entries of $\gamma$, the lemma follows.
\end{proof}

\begin{lemma} \label{lem:approximate_commutator}
Let $U\in U_{q^n}(\bbC)$ such that $D^{tr}(U,\cP)\ge \delta$ for some $\delta > 0$. Then,
\[
\Pr_{P\in\cP} \left(D^{tr}([U,P],\cI)\ge \frac{\delta^2}{2} \right) \ge \frac{\delta^2}{2}
\]
\end{lemma}

\begin{proof}
Let $U = \sum_{x,z}v_{x,z}X^xZ^z$ be the Pauli decomposition of $U$. 
The assumption on the distance of $U$ from $\cP$ implies that
\[\delta^2 \leq (D^{tr}(U,\cP))^2 = \min_{a,b}\left(1-\abs{\ntr(U(X^aZ^b)^{-1})}^2
\right) = 1 - \max_{a,b} \abs{v_{a,b}}^2.\]
Let $\epsilon_{a,b}^2 := D^{tr}([U,X^aZ^b],\cI)^2$.
Then, the expectation
\begin{eqnarray}
\bbE_{a,b}(1-\epsilon_{a,b}^2) &=& \bbE_{a,b} \left(\abs{\ntr([U,X^{a}Z^b])}^2\right) \\
&=& \bbE_{a,b}\left(\abs{\sum_{x,z} \abs{v_{x,z}}^2\cdot \omega_q^{az-bx}}^2\right)\\
&=& \label{distribution} \bbE_{a,b}\left(\abs{\bbE_{x,z}\; \w_q^{az-bx}}^2\right)\\
&=& \label{changing_Variables} \bbE_{a,b} \left(\bbE_{x,x',z,z'} \w_q^{az-bx-az'+bx'} \right) \\
&=& \label{sum_roots}\sum_{x=x',z=z'} \abs{v_{x,z}}^4 
\\
&\leq & \left(\max_{x,z} \abs{v_{x,z}}^2\right) \cdot \sum_{x,z} \abs{v_{x,z}}^2 \leq 1- \delta^2.
\end{eqnarray}
In equation~(\ref{distribution}), we view $\sum_{x,z} \abs{v_{x,z}}^2\cdot \omega_q^{az-bx}$ as the expectation of a distribution over the $q$-th roots of unity. In equation~(\ref{changing_Variables}), we treat $\abs{\bbE_{x,z}\; \w_q^{az-bx}}^2$ as multiplication of two expectation over independent variables and use the fact that $\sum_i \omega_q^i=0$, to derive equation~(\ref{sum_roots}). We conclude the proof by noting that the above implies that $\bbE(\epsilon_{a,b}^2)\ge \delta^2$. Therefore, by Markov's, we have that $\Pr_{P\in\cP} \left(D([U,P],\cI)\ge \frac{\delta^2}{2} \right) \ge \frac{\delta^2}{2}$.

\end{proof}

\section{Traces of Pauli and Clifford Operators are Discrete}
\label{sec:traces}

In this section we show that the trace of a Pauli or a Clifford operator over prime qudits cannot take arbitrary values. Rather, it can only take one of a few discrete values. This is in contrast to general unitaries which can have arbitrary trace (e.g.\ single-qubit rotation matrix). 
We start with a general claim about the trace of identity.
\begin{lemma}\label{lem:trid}
Let $U$ be a unitary operator over $n$ $q$-ary qudits. Then, $\abs{\ntr(U)}\le 1$, and equality holds 
if and only if $U \in \cI$.
\end{lemma}

\begin{proof}
Recall that $U$ is square matrix of dimension $q^n$, and that the trace is invariant under basis change. Therefore, $\tr(U)=\sum_x \lambda_x$, where $\lambda_x$ denotes the eigenvalues of $U$. Since $U$ is unitary, $\abs{\lambda_x}=1$, for all $x$. We now use triangle inequality to conclude that 
\[\abs{\tr(U)} = \big| {\sum_x \lambda_x} \big| \le \sum_x \abs{\lambda_x} = q^n~.\]
Equality holds only when all $\lambda_x$'s are equal, i.e., when 
$U=\lambda I$, where $\abs{\lambda}=1$.
\end{proof}

\subsection{Trace of Pauli Matrices}

In the case of Pauli operators, all non-identity group elements 
have zero trace.
\begin{lemma}[Trace of Pauli is Zero/One]\label{lem:trPauli}
Let $P \in \cI \cP$, i.e.,\  $P$ is equivalent to a Pauli operator up to a global phase. Then, 
$\abs{\ntr(P)} = 1$ if and only if $P \in \cI$, and $\abs{\ntr(P)} = 0$, otherwise.
\end{lemma}
\begin{proof}
This follows from the fact that all non-identity single-qudit Paulis have trace $0$. Since the Pauli group is a tensor of single-qudit Paulis, the result follows.
\end{proof}

\subsection{Trace of Clifford Matrices}

We show that for any prime $q$, the trace of a $q$-ary Clifford unitary must be a power of $\sqrt{q}$, or $0$. 
\begin{theorem}[\textit{Cliffords have Discrete Trace}]\label{thm:cliffordtraceoddq}
For a Clifford unitary $C$ over $n$ 
qudits, it holds that $\abs{\tr(C)}$ (and thus also $\abs{\ntr(C)}$) are either $0$ or a power of $\sqrt{q}$.
\end{theorem}

The following is an immediate corollary of Theorem~\ref{thm:cliffordtraceoddq}, combined with Lemma~\ref{lem:trid}.
\begin{corollary}\label{cor:trclif}
Let $C \in \cI \cC$, i.e.,\  $C$ is equivalent to a Clifford operator up to a global phase. Then $\abs{\ntr(C)} = 1$ if and only if $C \in \cI$, and $\abs{\ntr(C)} \le 1/\sqrt{q}$, otherwise.
\end{corollary}

To prove Theorem~\ref{thm:cliffordtraceoddq}, we use certain results 
on quadratic exponential sums and argue the desired claims using modular arithmetic. We analyze the case of $q=2$ separately from that of $q$ being an odd prime. The proof is a novel contribution of this work. However, it is independent of the results in the following sections and also somewhat lengthy, we include it in Appendix~\ref{apx:traces}.

\section{The Subgroup Testing Problem}\label{sec:subgroup}

\def\itp{\mathsf{ITP}}
\def\ptp{\mathsf{PTP}}
\def\ctp{\mathsf{CTP}}

Given a subgroup $\cG$ of the unitary group and a quantum unitary operator $C$, we consider the problem of testing, whether $C$ implements a unitary that is ``close'' to $\cG$, or ``far'' from it, according to a specified distance measure $D$. We also consider a promise version of this problem, wherein, given subgroups $\cG \subseteq \cG'$ of the unitary group, and a circuit $C$ with a promise that $C$ is either ``close'' or ``far'' from $\cG'$, we decide whether $C$ is ``close'' or ``far'' from the smaller subgroup $\cG$. Formal definitions follow.

\begin{definition} [$(\cG_{a,b},D)$-testing]
Let $\cG$ be a subgroup of $U_{q^n}(\C)$. Given a quantum unitary operator $C$ that acts on $n$ qudits and a distance measure $D$, the problem of $(\cG_{a,b},D)$-testing is to decide whether $C$ satisfies $D(C, \cG) \leq a$ (YES) or $D(C, \cG) \geq b$ (NO), 
assumed one of these to be the case. 
\end{definition}

\begin{definition} [$(\cG_{a,b},\cG'_{a',b'},D)$-testing]
Let $\cG \subseteq \cG'$ be subgroups of $U_{q^n}(\C)$. Given a quantum unitary operator $C$ that acts on $n$ qudits with the promise that $C$ satisfies $(\cG'_{a',b'},D)$-testing with a distance measure $D$, the problem of $(\cG_{a,b},\cG'_{a',b'},D)$-testing is to decide whether $C$ satisfies $D(C, \cG) \leq a$ (YES) or $D(C, \cG) \geq b$ (NO),
assumed one of these to be the case.
\end{definition}
We specialize to the case where $a=0$ and the parameter $b$ is the whole interval $(0,1]$ and call them the exact problems. Observe that this version is independent of any distance measure. Formal definitions follow.
\begin{definition}[Exact $\cG$-testing and $(\cG,\cG')$-testing problems]\label{def:exact_problem} Let $\cG$, $\cG'$ be a subgroup of $U_{q^n}(\C)$ such that $\cG \subseteq \cG'$. Given a quantum unitary operator $C$ that acts on $n$ qudits, the exact $\cG$-testing problem is the problem of deciding, whether $C$ implements a unitary that belong to $\cG$.

The promise $(\cG,\cG')$-testing problem is the problem of $\cG$-testing, under the promise that $C$ implements a unitary from the bigger group $\cG'$.
\end{definition}
In the setup of \textit{property testing}, the parameter $a$ equals $0$, whereas $b$ is a fixed real number in the interval $(0,1]$. 
The most general version, i.e., $0 < a \leq b$ is called \textit{tolerant testing}.

In the following result, we show how, in either of the setups - exact, property or tolerant testing, the promise $(\cG,\cG')$-testing problem can be used to reduce $\cG$-testing to $\cG'$-testing. 
\begin{definition}[\textbf{A Composition Algorithm}.]\label{algo:composition}
Given subgroups $\cG \subseteq \cG'$ and algorithms for solving $(\cG'_{a,b},D)$-testing and $(\cG_{a,b},\cG'_{a,b},D)$-testing, we construct an algorithm for the problem of $(\cG_{a,b},D)$-testing, for the smaller group $\cG$, as follows. 

Let $T_{\cG'}$ and $T_{(\cG', \cG)}$ denote the $(\cG'_{a,b},D)$-testing and $(\cG_{a,b},\cG'_{a,b},D)$-testing algorithms, respectively. Let $C$ be an input quantum unitary operator. We propose the following unconditional test for $\cG$: apply the two tests $T_{\cG'}$ and $T_{\cG',\cG}$ on $C$, and accept if both the tests accept, reject if $T_{\cG'}$ rejects and abort, otherwise. The correctness of the algorithm is justified as follows:
\begin{itemize}
    \item If $T_{\cG'}$ answers that $C$ is a NO instance for $(\cG'_{a,b},D)$, then $C$ is a NO instance for $(\cG_{a,b},D)$ as well.
    \item If both $T_{\cG'}$ and $T_{(\cG', \cG)}$ answer that $C$ is a YES instance for $(\cG'_{a,b},D)$ and $(\cG_{a,b},\cG'_{a,b},D)$, respectively. Then, indeed $D(C, \cG) \leq a$. 
\end{itemize}

\end{definition}

\paragraph{Identity-Testing, Pauli-Testing, Clifford-Testing.} 

\def \tp {\mathsf{TP}}
Given unitary operators on the space of 
$q$-ary qudits, the identity-testing problem $(a,b, D)$-$\itp$ is the problem of $(\cG_{a,b},D)$-testing with $\cG = \cI$. The Pauli-testing problem $(a,b, D)$-$\ptp$ is the $(\cG_{a,b},D)$-testing problem with $\cG = \cI \cP$. The Clifford testing problem $(a,b, D)$-$\ctp$ is the $(\cG_{a,b},D)$-testing problem with $\cG= \cI \cC$.

\paragraph{The Representation of the Input $C$.} Among the 
various ways of representing the input $C$ to the testing problem, we consider the minimal one being oracle access and the maximal being given a circuit that computes $C$. 
We refer to the former as black-box (BB) and to the latter as white-box (WB). 
We often relax the notion of BB access and consider, in addition to a $C$ oracle, also access to a $C^\dagger$ oracle.
\subsection{Promise Identity-Testing for Pauli and Clifford}
We consider the exact versions, i.e., $\cG$-testing, of the Identity, Pauli and Clifford testing problems and give efficient tests for their promised versions. 
In Section~\ref{sec:equivalences}, we use these tests to construct algorithms for the $(\cG_{a,b}, \cG'_{a,b},D)$-testing problems.
\begin{theorem}[Black-Box Identity-Testing for the Pauli and Clifford Groups]\label{thm:BBTest}
The EPR Test (Test~\ref{test:eprtest}) is a test for $\itp$ with the following promise:
\begin{itemize}
    \item[(i)] if the input is in $\mathcal{IP}$, the test has both perfect completeness and prefect soundness, whereas
    \item[(ii)] if the input is in $\mathcal{IC}$, then the test has perfect completeness and soundness bounded by $1/q$. 
\end{itemize}
\end{theorem} 

\begin{proof}
By Lemma~\ref{lem:eprtest}, Test~\ref{test:eprtest} accepts with probability $\abs{\ntr(C)}^2$. If $C \in \cI$ then $\abs{\ntr(C)}=1$ and thus completness is perfect by Lemma~\ref{lem:trid}. By the same lemma, if $C \not\in\cI$ then $\abs{\ntr(C)}<1$. If $C \in \cI\cP \setminus \cI$ then by Lemma~\ref{lem:trPauli} it holds that $\abs{\ntr(C)}=0$, thus $(i)$ follows. If $C \in \cI\cC \setminus \cI$ then by Theorem~\ref{thm:cliffordtraceoddq} $\abs{\ntr(C)} \le 1/\sqrt{q}$ and thus $(ii)$ follows.
\end{proof}

Our WB test takes as input 
a circuit that contains only Clifford gates, and checks whether this circuit implements the identity circuit. Note that for the ``standard'' binary Cliffords, this is immediate from the existence of a canonical representation, and in fact one can prove canonical representation for $q$-Cliffords and derive a WB test from there. However, we take a more direct route which relies on the definition of the Clifford group as the normalizer of the Pauli group.
\begin{test}[WB Identity-Testing for Clifford]\label{test:WB} Given a circuit $C$ consisting of Clifford gates only, compute the following (vectors of) polynomials $f$, $g$ and $h$, 
that arise from conjugating by $C$,
\[C X^a Z^b C^\dagger = \w_{q'}^{h(a,b)} X^{f(a,b)} Z^{g(a,b)}.\]
Here $X^a Z^b$ is the formal Pauli element $(\otimes_i X_q^{a_i}Z_q^{b_i})$,  and the vectors of polynomials $f(a,b) = (f_i(a,b))_i$, $g(a,b) = (g_i(a,b))_i$ and  
$h$ are computed as per Theorem~\ref{thm:clifconjugate}. Accept if and only if $f(a,b) = a$, $g(a,b) = b$ as vectors 
and $h$ is the zero polynomial.
\end{test}

\begin{theorem}[White-box Identity-Testing for the Clifford group]\label{thm:WBTest} Algorithm~\ref{test:WB} has perfect completeness and soundness.
\end{theorem}

\begin{proof}
If $C \in \cI$ then by definition it leaves all Paulis unchanged, which translates to $f$, $g$ satisfying $f(a,b) =a$, and $g(a,b) =b$, 
and $h$ being the zero polynomial. Thus completeness is perfect. On the other hand, if $f$, $g$ satisfy $f(a,b) =a$, and $g(a,b) =b$, respectively, 
and $h$ is zero, then for any Pauli $P$, the commutator $[C,P]=CPC^\dagger P^\dagger=I$. Since the Paulis generate the set of complex matrices (of the same dimension) as a vector space (Theorem~\ref{thm:Pauli_decomp}), then $C$ commutes with all matrices and is therefore a scalar matrix. 
Further, as $C$ is unitary, $C \in \cI$.
\end{proof}

\section{Equivalences Between Subgroup Testing Problems}\label{sec:equivalences}

\def\itp{\mathsf{ITP}}
\def\ptp{\mathsf{PTP}}
\def\ctp{\mathsf{CTP}}
\def \A{\widetilde{A}}
In this section, we show various randomized reductions between the three aforementioned problems, Identity, Pauli and Clifford testing. Our first two reductions, i.e., $(a,b,D)$-$\itp$ to $(a,b,D)$-$\ptp$ and $(a,b,D)$-$\itp$ to $(a,b,D)$-$\ctp$ follow from the composition algorithm (definition~\ref{algo:composition}), and hold for all three versions; exact, property and tolerant testing. Our $(a,b,D)$-$\ctp$ to $(a',b',D)$-$\ptp$ reduction is an adaptation of known work. See Subsection~\ref{subsec:ctp_ptp}, for more details. The final reduction, from $(a,b,D)$-$\ptp$ to $(a,b,D)$-$\itp$, in Subsection~\ref{subsec:ptp_itp} uses our analysis, from Section~\ref{sec:commute}, of the fixed points under the conjugation action of a unitary $C$ on the Paulis. For the last two reductions, we restrict to either the exact problem or the setup of property testing, i.e., $a=0$. 
Finally, in Section~\ref{subsec:complexity}, we use the reductions to derive several conclusions on the complexity of these problems.

\subsection{$(a,b,D)$-$\itp \leq (a,b,D)$-$\ptp$}
\label{subsec:ITP-PTP/CTP}

Our composition algorithm (definition~\ref{algo:composition}) immediately implies the following reduction from $(a,b, D)$-$\itp \le (a,b,D)$-$\ptp$: Given a unitary operator $C$ and a solver for $(a,b,D)$-$\ptp$, the algorithm $A_{\mathsf{ItoP}}$ runs the test described in the composition algorithm using the $(a,b,D)$-$\ptp$-solver and the BB test (Theorem~\ref{thm:BBTest}) for the promise $(\cI \cP, \cI)$-testing problem to construct a solver for $(a,b,D)$-$\itp$. 
    
\begin{proposition}\label{reduction:identity->pauli}
Assume that the $(a,b)$-$\ptp$ is perfectly complete and sound. Then,
\begin{itemize}
    \item[i.] When $a=0$, and $D \in \{D^{tr}, D^{op}\}$, the $(a,b,D)$-$\itp$ solver is perfectly complete and sound.
    \item[ii.] When $a>0$, and $D= D^{op}$, the $(a,b,D)$-$\itp$ solver has completeness and soundness bounded by $(1-a)^2$ and $a^2$, respectively.
    \item[iii.] When $a>0$, and $D=D^{tr}$, the $(a,b,D)$-$\itp$ solver has completeness and soundness bounded by $1-a^2$ and $a^2$, respectively.
\end{itemize}
\end{proposition}

\begin{proof}
Recall that $A_{\mathsf{ItoP}}$ accepts when both the $(a,b)$-$\ptp$ and the BB test accept. Under the assumption that the $(a,b)$-$\ptp$ solver is perfectly complete and sound, the completeness of the $(a,b)$-$\itp$ tester is at least $\abs{\ntr(C)}^2$, which is the acceptance probability of the BB test, whereas the soundness is bounded above by $\abs{\ntr(C)}^2$ (see lemma \ref{lem:eprtest}). We bound $\abs{\ntr(C)}^2$ in the three cases.
\begin{itemize}
    \item[i.] When $a=0$, the unitary operator $C$ is either in $\cI$, in which case, $\abs{\ntr(C)}=1$, or it is in $\cP \setminus \cI$, in which case $\abs{\ntr(C)}=0$, by Lemma~\ref{lem:trPauli}.
    \item[ii.] Let $a>0$ and $D=D^{op}$. For arguing completeness, let $D(C, \cI) \leq a$. Then $C$ should be accepted by the $\itp$, and indeed, it is accepted with probability at least $\abs{\ntr(C)}^2$. In order to calculate the trace, recall that $\abs{\ntr(C-I)} \leq \norm{C-I}_{op}$ (by Lemma~\ref{lem:matrixnorms}), and therefore
    \[1 = \abs{\ntr(I)} \leq \abs{\ntr(C-I)} + \abs{\ntr(C)} \leq a + \abs{\ntr(C)} \]
    implies that $\abs{\ntr(C)}^2\le (1-a)^2$.\\
    For soundness, we assume that $D(C,\cI) \geq b$, in which case the $\itp$ should not accept it. However, the $\itp$ does accept it when the $\ptp$ accepts it (i.e., $D(C, \cP) \leq a$) and the BB test accepts it. The latter accepts it with probability bounded above by $\abs{\ntr(C)}^2$. Therefore, soundness of the $\itp$ is bounded above by 
    \[ \abs{\ntr(C)}^2 = \abs{\ntr(C-P)}^2 \le \norm{C-P}_{op}^2 \le a^2 \]
    \item[iii.] Let $a>0$ and $D= D^{tr}$. We argue for completeness and soundness exactly as in ii. For completeness, we assume that $D(C, \cI) \leq a$. Then, by definition, $\abs{\ntr(C)}^2 \geq 1-a^2$, and we get that completeness is at least $1-a^2$. For soundness, we assume that $D(C,\cI) \geq b$ and $D(C, P) \leq a$. The latter implies that $1-\abs{\ntr(CP^\dagger)}^2 \leq a^2$. 
    Observe that $\ntr(CP^{\dagger})$ is the coefficient of $P$ in the Pauli decomposition of $C$. Since $C$ is a unitary, sum of the squares of the absolute value of its coefficients from the Pauli decomposition should equal $1$. Therefore, $\abs{\ntr(C)}^2 + \abs{\ntr(CP^\dagger)}^2 \leq 1$, and we get that $\abs{\ntr(C)}^2 \leq a^2$.
\end{itemize}
\end{proof}
    
\subsection{$(a,b,D)$-$\itp \leq (a,b,D)$-$\ctp$}
Another application of the composition algorithm (definition~\ref{algo:composition}) is a reduction from $(a,b,D)$-$\itp$ to $(a,b,D)$-$\ctp$. Given unitary operator $C$ and a solver for $(a,b,D)$-$\ctp$, algorithm $A_{\mathsf{ItoC}}$ runs the test described in the composition algorithm using the $(a,b,D)$-$\ctp$-solver and 
    the BB test (Theorem~\ref{thm:BBTest}) for the promise $(\cI \cC, \cI)$-testing problem, to construct a solver for $(a,b,D)$-$\itp$. 
    Assuming that the $(a,b,D)$-$\ctp$ solver has perfect completeness and soundness, we compute the completeness and soundness of the $(a,b,D)$-$\itp$ solver as follows.

\begin{proposition} \label{reduction:identity->clifford}
Assume that the $(a,b, D)$-$\ctp$ solver is perfectly complete and sound. Then,
\begin{itemize}
    \item[i.] When $a=0$, and $D \in \{D^{tr}, D^{op}\}$, the $(a,b,D)$-$\itp$ solver is perfectly complete and has soundness at most $\frac{1}{q}$.
    \item[ii.] When $a>0$, and $D= D^{op}$, the $(a,b,D)$-$\itp$ solver has completeness and soundness bounded by $(1-a)^2$ and $(\frac{1}{q}+a)^2$, respectively.
    \item[iii.] When $a>0$, and $D^{tr}$, the $(a,b,D)$-$\itp$ solver has completeness and soundness bounded by $1-a^2$ and $1- (\sqrt{1-1/q}-a)^2$, respectively.
\end{itemize}
\end{proposition}

\begin{proof}
Let $C$ be the input unitary operator, described above. The blueprint of the proof and the completeness arguments are identical to that of the proof of Proposition~\ref{reduction:identity->pauli}. For soundness, recall that we assume that $D(C,\cI) \geq b$ and $D(C, \cC) \leq a$. Then, the soundness of the $\itp$ is bounded above by $\abs{\ntr(C)}^2$.
\begin{itemize}
    \item[i.] When $a=0$, and the unitary operator $C$ is in $\cC \setminus \cI$, then $\abs{\ntr(C)}^2\leq \frac{1}{q}$.
    \item[ii.] When $a>0$, and $D = D^{op}$. If $D(C,\cI) \geq b$, and $D(C,\cC) \leq a$ then $\abs{\ntr(C- C')} \leq \norm{C- C'}_{op}$ (by Lemma~\ref{lem:matrixnorms}), for some $C' \in \cC$, implies that soundness is at most $a + \frac{1}{q}$, since
    \[\abs{\ntr(C)} \leq \abs{\ntr(C-C')} + \abs{\ntr(C')} \leq a + \frac{1}{q}\]
    \item[iii.] When $a>0$, and $D=D^{tr}$. If $D(C, \cI) \geq b$ and $D(C, \cC) \leq a$, we use the definition of $D^{tr}$ and the fact that $D^{tr}(C,\cI) \leq D^{tr}(C,C') + D^{tr}(C', \cI)$ to conclude that $\abs{\ntr(C)}^2 \leq 1 - \left(\sqrt{1 - \frac{1}{q}} -a \right)^2$, in this case.
\end{itemize}
\end{proof}

\subsection{$(0,b,D^{tr})$-$\ctp \leq (0,b', D^{tr})$-$\ptp$}\label{subsec:ctp_ptp}
The following reduction is a generalization of the results in \cite[Section V]{wangPropertyTesting} to the $q$-ary setting. Even though the author restricts to working with qubits, the proofs hold for qudits as well. Further, the author uses the distance measure $D(U,V) = \sqrt{1- \abs{\ntr(UV^{\dagger})}}$, which is equivalent to our distance measure $D^{tr}$. See~\cite[Section 5.1.1, pg 46]{montanaro2018survey} for an explanation of the equivalence of the two distance measures. We state our generalization of the result here and refer the reader to \cite[Section V]{wangPropertyTesting} for a proof. 

Given a unitary operator $C$ and a $(0,b', D^{tr})$-$\ptp$ solver, the reduction algorithm $A_{\mathsf{CtoP}}$ runs the following test $O(1)$ number of times; pick a random Pauli $P \in \cP$, and run $CPC^{\dagger}$ through the $(0, b', D^{tr})$-$\ptp$ solver and output YES if all the iterations accept, otherwise output NO.
\begin{proposition}\cite[Algorithm 4]{wangPropertyTesting} \label{reduction:clifford->pauli}
Assume that the $(0,\frac{b}{8}, D^{tr})$-$\ptp$ solver is perfectly complete and sound, then $(0, b, D^{tr})$-$\ctp$ 
is perfectly complete with soundness bounded above by $\frac{1}{3}$.
\end{proposition} 

\begin{remark}
In \ref{subsec:ctp_ptp}, it is assumed that the reduction receives $C^\dagger$ in addition to $C$. The recent work of \cite{Gross_2021} improves upon \cite{wangPropertyTesting} by removing the need for $C^{\dagger}$ when testing if $C$ is a Clifford in a black box manner. Our reduction however uses the inverse $C^\dagger$ to reduce the Clifford testing problem to Pauli testing problem, so we keep the reduction from \cite{wangPropertyTesting}.  
\end{remark}

\begin{remark}
For the exact version of the problem, one could construct, as follows, a $\cC$-tester that is perfectly complete and perfectly sound, as well. 
Let $P^x_i=I^{\otimes i}\otimes X_q\otimes I^{\otimes n-i-1}$ and $P^z_i=I^{\otimes i}\otimes Z_q\otimes I^{\otimes n-i-1}$, for $i\in [n]$. Given a unitary operator over $n$ qudits and $\cP$-tester (definition~\ref{def:exact_problem}), the algorithm $A_{\mathsf{CtoP}}$ queries the $\cP$-tester on each of the $2n$ matrices $\{C^{\dagger}P_i^xC, C^\dagger P_i^xC\}_{i \in[n]}$, and accept if and only if the $\cP$-tester accepts on all. Correctness follows by the fact that the set $\{P_i^x, P_i^z\}_{i\in[n]}$ generates the Pauli group $\mathcal{P}_q^n$ up to global phases, and the fact that Cliffords normalize the Pauli group. 
\end{remark}

\subsection{$(0,b,D^{tr})$-$\ptp \leq (0,b', D^{tr})$-$\itp$}\label{subsec:ptp_itp}
Given an $(0,b', D^{tr})$-$\itp$ solver, the reduction  algorithm $A_{\mathsf{PtoI}}$ creates the following circuit $C'$, sends $C'$ to the $(0,b', D^{tr})$-$\itp$ solver and answers the same as the $(0,b', D^{tr})$-$\itp$ solver does. To generate $C'$, the algorithm $A_{\mathsf{PtoI}}$ samples a random $P \in \cP_q^n$, and defines:
\[ C' = [C,P] := C^{\dagger}P^{\dagger}CP\]
\begin{proposition} \label{reduction:pauli->identity} 
Assume that the $(0,\frac{b^2}{2}, D^{tr})$-$\itp$ solver is perfectly complete and sound. Then, $(0,b,D^{tr})$-$\ptp$ has perfect completeness and soundness bounded above by $\min \left\{\frac{1}{q},\left(1-\frac{b^2}{2}\right)\right\}$. 
\end{proposition}

\begin{proof}
When $C \in \cP$, then, by definition, $[C,P] \in \cI$, for all $P \in \cP$, and therefore the $\ptp$ always accepts and is perfectly complete. When $C \notin \cP$, by Lemma~\ref{lem:approximate_commutator}, with probability at most $\left(1-\frac{b^2}{2}\right)$, the distance $D^{tr}(C',\cI) \leq \frac{b^2}{2}$. Further, in this case, since $C \notin \cP$, at most $\frac{1}{q}$ fraction of the Paulis satisfy $[C,P] \in \cI$, by Lemma~\ref{lem:pauli_commute}. The result follows.
\end{proof}

\subsection{Complexity Implications}\label{subsec:complexity}
The reductions between the three avatars of Identity, Pauli and Clifford testing problems; exact, property and tolerant testing yield the following corollaries. We include a short description of the proofs at the end of this section.
\begin{corollary}
The exact problem (see definition~\ref{def:exact_problem}) of $\cI$-testing, $\cP$-testing and $\cC$-testing are all equivalent. Further, if the input is provided using an efficient implementation, the $\cP$-testing and $\cC$-testing problems are co-NQP complete under randomized reductions.
\end{corollary}

\begin{corollary}
The problems $(a,b, D^{op})$-$\ptp$ and $(a,b, D^{op})$-$\ctp$, when the input is provided using an efficient implementation, are co-QMA hard (QMA hard under Turing reductions) for $b-a\geq 1/poly(n \log q)$.
\end{corollary}

\begin{corollary}
Under the property testing setup, we have the following randomized reductions with distance measure $D^{tr}$,
\begin{itemize}
    \item $(0,b,D^{tr})-\itp \lra (0,b,D^{tr})- \ctp \lra (0,b/8, D^{tr})- \ptp \lra (0, b^2/128, D^{tr})- \itp$\\ 
    \item $(0,b,D^{tr})-\itp \lra (0,b, D^{tr})- \ptp \lra (0, b^2/2, D^{tr})- \itp$
\end{itemize}
\end{corollary}

The first corollary follows from Proposition~\ref{reduction:identity->pauli}$i.$ and Proposition~\ref{reduction:identity->clifford}$i.$ and from the fact that the non Identity testing problem is NQP complete (\cite[Theorem 1]{tanaka2009exact}). The second corollary follows from Proposition~\ref{reduction:identity->pauli}$ii.$ and Proposition~\ref{reduction:identity->clifford}$ii.$ and the fact that non-identity test is QMA-complete~\cite{NonIDtesting}. Both assume that if the input to $\ptp/\ctp$ is given in an efficiently executable manner, then in particular it is efficient to implement the solvers that are required by our reductions. The final corollary follows from Proposition~\ref{reduction:identity->pauli}$iii.$, Proposition~\ref{reduction:identity->clifford}$iii.$, Proposition~\ref{reduction:clifford->pauli} and Proposition~\ref{reduction:pauli->identity}.

\bibliographystyle{alpha}
\bibliography{biblio}
\appendix
\section{Appendix: Traces of Pauli and Clifford Operators are Discrete - Deferred proof of Theorem~\ref{thm:cliffordtraceoddq}}\label{apx:traces}
We begin with certain definitions and results on quadratic exponential sums that we use to prove Theorem~\ref{thm:cliffordtraceoddq}.
\subsection{Sum-Over-Paths Formalism} 
\label{sec:prelim:sumoverpaths}
The sum-over-paths formalism provides a way to express the outcome of the action of a quantum circuit on qudits. We restrict to Clifford circuits for our purpose. See \cite{Koh2017ComputingQC,Wolf2019QuantumCL} for further details. 

Let $C$ be a $q$-ary Clifford circuit that is classically described as an ordered sequence of the generators of $\mathcal{C}_q^n$ and which registers they act on. Under the sum-over-paths technique, one follows the circuit, one gate at a time, and updates the amplitudes of the computational basis elements according to the definition of the gate. We refer the reader to \cite[Sec III]{Koh2017ComputingQC} for a detailed explanation. 
We assume without loss of generality that $C$ has $k$ intermediate $\bbF_q$ gates and $n$ terminal $\bbF_q$ gates, the latter applied to each output wire. 

\begin{proposition}\label{prop:clifequation}
Let $C$ be a Clifford circuit acting on $n$ $q$-ary qudits, with $k$ intermediate and $n$ terminal $F_q$ gates. Then in the binary and odd-prime settings, respectively, 
\begin{align}\label{eq:phasepol2}
C\ket{x}&= \frac{1}{2^{(k+n)/2}}  \sum_{y \in \bbF_2^n} \ket{y} \sum_{v\in \bbF_2^k} \w_4^{s(x,v)}\w_2^{g(x,v,y)} = \frac{1}{2^{(k+n)/2}}  \sum_{y \in \bbF_2^n} \ket{y} \sum_{v\in \bbF_2^k} \w_4^{s(x,v)+2g(x,v,y)} \\
\label{eq:phasepolq}
C\ket{x}&= \frac{1}{q^{(k+n)/2}}\sum_{y \in \bbF_q^n} \ket{y} \sum_{v \in \bbF_q^k} \w_q^{h(x,v,y)}, 
\end{align}
where $s$ is a linear polynomial and $g,h$ are quadratic polynomials.
\end{proposition}

The case of odd prime $q$ has been analyzed in \cite{Koh2017ComputingQC} and the case of $q=2$ is fairly similar but we provide a proof for the sake of completeness.

\begin{proof}
We will show that without the terminal $F$ gates it holds that
\begin{align}\label{eq:intemediatesumoverpathform2}
C\ket{x}&= \frac{1}{2^{k/2}}   \sum_{v\in \bbF_2^k} \w_4^{s(x,v)}\w_2^{g(x,v)} \ket{\ell(x,v)} \\
\label{eq:intemediatesumoverpathformq} C\ket{x}&= \frac{1}{q^{k/2}} \sum_{v \in \bbF_q^k} \w_q^{h(x,v)} \ket{\ell(x,v)}, 
\end{align}
where $\ell$ are linear functions. Applying the terminal $F$ gates then immediately implies Eq.~\eqref{eq:phasepol2},~\eqref{eq:phasepolq}. Eq.~\eqref{eq:intemediatesumoverpathform2},~\eqref{eq:intemediatesumoverpathformq} simply follow by induction on the functionality of the Clifford gates as described above.
\end{proof}

\paragraph{Respectful Polynomials.} For our analysis of the trace of binary Clifford circuits (Section~\ref{sec:traces}), we will require the following definition.

\begin{definition}\label{def:respectpol}
	Let $h(x)$ be a quadratic polynomial with integer coefficients over variables $x_1, \ldots, x_n$. We say that $h$ is \emph{respectful} if it can be written as $h(x)=s(x)+2g(x)$, where $s(x)$ is a linear function in the squares of the variables, i.e.\ in $x^2_1, \ldots, x^2_n$ and $g(x)$ is an arbitrary quadratic polynomial.
\end{definition}

The following proposition asserts that the sum-over-paths expression for binary Cliffords yields a respectful polynomial in the exponent. 
\begin{proposition}
\label{prop:respectpol}
The polynomial obtained in Eq.~\eqref{eq:phasepol2}, in the binary case, can be rewritten as $\w_4^{h'(x,v,y)}$ where $h'$ is respectful. \end{proposition}

\begin{proof}
From Eq.~\eqref{eq:phasepol2} we have, in the exponent of $\omega_4$, the polynomial $s(x,v)+2g(x,v,y)$, where $s$ is linear and $g$ is quadaratic. However, since we only sum over binary variables for which $z^2=z$, we can substitute each variable in $s$ by its square and obtain an equvialent respectful polynomial.
\end{proof}

\subsection{Proof of Theorem \ref{thm:cliffordtraceoddq}}\label{apx:proofofcliffordtrace}
We analyze the case of $q=2$ separately from that of $q$ being an odd prime. 
We begin with the latter which is fairly straightforward.

\begin{lemma}\label{lem:sumovervars}\label{lem:clifq}
Let $q$ be an odd prime. Let $h(x)$ be an multivariate quadratic polynomial  over variables $x_1, \ldots, x_n$. Then there exist an integer $t\le n$, a root of unity $c$ and $z \in \binset$
 such that
\begin{align}\label{eq:expsumodd}
    \sum_{x}\omega_q^{h(x)}= q^{\frac{t}{2}} \cdot c \cdot z~.
\end{align}
\end{lemma}
\begin{proof}
When $h(x)$ is a polynomial in one variable (i.e, $n=1$), the equality in Equation~(\ref{eq:expsumodd}) follows from ~\cite[Theorem 5.33]{lidl_niederreiter_1996}. 
For the multi-variate setting, we note that any quadratic form over an odd prime order field can be diagonalized, thus via change of variables, $\sum_{x}\omega_q^{h(x)}$ can be expressed as a product of sums of single-variable quadratic functions. 
\end{proof}

\begin{lemma}[The $q=2$ case] \label{lem:clliford2}

	Let $h(x)$ be a respectful quadratic polynomial over variables $x_1, \ldots, x_n$, as defined in definition~\ref{def:respectpol}.
	Then there exist an integer $t$, a root of unity $c$ and $z\in\binset$ such that
	\begin{align}
		\sum_{x \in \binset^n} \omega_4^{h(x)} = 2^{\frac{t}{2}} \cdot c  \cdot z ~.
	\end{align}
\end{lemma}

We start with the proof of Theorem~\ref{thm:cliffordtraceoddq} based on the two Lemmas from above, followed by a proof for Lemma~\label{lem:clif2} which is a bit more involved. 

\begin{proof}[Proof of Theorem~\ref{thm:cliffordtraceoddq}]
Let $C$ be a Clifford circuit acting on $n$ $q$-ary qudits, with $k$ intermediate and $n$
terminal $F_q$ gates. Consider the sum-over-paths formalism, as described in Section~\ref{sec:prelim:sumoverpaths}. By Propositions~\ref{prop:clifequation} and~\ref{prop:respectpol}, and by definition of the trace, we have
\begin{align}
\tr(C) = \sum_x \bra{x}C\ket{x}=\frac{1}{q^{(n+k)/2}}\sum_{x,v}\omega_{q'}^{g(v,x)}~,
\end{align}
where if $q=2$ then $q'=4$ and $g$ is respectful, and if $q$ is an odd prime then $q'=q$ and $g$ is an arbitrary quadratic polynomial over $k+n$. Applying Lemma~\ref{lem:clliford2} or Lemma~\ref{lem:clifq} for the respective cases, the theorem follows.
\end{proof}

\paragraph{Proof of Lemma~\ref{lem:clliford2}}
\label{sec:proofclif2}

We start by establishing some useful properties of respectful polynomials.

\begin{claim}\label{claim:respectfullinear}
    The property of being respectful is preserved under integer linear transformations of the variables.
\end{claim}
\begin{proof}
This follows since $(x+y)^2=x^2+2xy+y^2$, so every monomial that is not a square of a variable has an even coefficient.
\end{proof}

\begin{claim}\label{claim:respectmod2}
	If $h$ is respectful then for any integer vector $x$ it holds that $h(x) = h(x \pmod{2}) \pmod{4}$. 
\end{claim}
\begin{proof}
	The claim follows since for all integer $a,b$ it holds that $(a+2b)^2=a^2 \pmod{4}$ and $2a = 2(a+2b) \pmod{4}$.
\end{proof}

The following is derived by direct calculation.

\begin{claim}\label{claim:onesquare}
	For any integer $y$, $\sum_{x \in \bbZ_4} \omega_4^{(x+y)^2} = 2(1+i)=2 \sqrt{2} \cdot \zeta$ where $\zeta = \tfrac{i+1}{\sqrt{2}}$ is a root of unity.
\end{claim}

\begin{lemma}\label{lem:squarecase}
For any integer $y$, it holds that 
\[\sum_{x \in \binset} \omega_4^{x^2+2xy} = \tfrac{1}{2} \sum_{x \in \bbZ_4} \omega_4^{x^2+2xy} = \tfrac{1}{2} \sum_{x \in \bbZ_4} \omega_4^{(x+y)^2-y^2} =\sqrt{2} \zeta \omega_4^{-y^2}~.\]
\end{lemma}
Note that we changed from summing over binary values to summing over values in $\bbZ_4$.
\begin{proof}
The first equality follows from Claim~\ref{claim:respectmod2}. The second is opening parenthesis. The third is an application of Claim~\ref{claim:onesquare}.
\end{proof}

The following is derived by direct calculation.
\begin{claim}\label{claim:linearcase}
	Let $x$ be a (scalar) variable and let $y$ be a $k$-dimensional vector of variables. Let $\ell(y)$ be a linear function and $g(y)$ be an arbitrary function. Then,
	\begin{align*}
		\sum_{x \in \binset} \omega_4^{2 x \ell(y) + g(y)} &= 2 \omega_4^{g(y)} \delta_{(\ell(y)\bmod{2})} ~.
	\end{align*}
Namely, the expression is $0$ if $\ell(y) \neq 0 \mod{2}$ and is equal to $2 \omega_4^{g(y)}$, otherwise. 
\end{claim}

\begin{proof}[Proof of Lemma~\ref{lem:clliford2}]
	We induct  
	on $n$. The base case of $n=0$ (i.e., $h$ is a constant) follows from definition. For general $n$, assume w.l.o.g that $h$ is not degenerate in $x_n$, i.e., $x_n$ appears in the expression. Otherwise, the claim just follows from the induction hypothesis. 
	Write $$h(x)=c_n x_n^2 + 2 x_n \ell(x') + g(x')$$ where $x'=(x_1, \ldots, x_{n-1})$, $\ell$ is a linear function over $\bbZ_2$ and $g$ is a respectful quadratic polynomial. 
	Note that the definition of a respectful polynomial allows us to assume that 
	$c_n \neq 2$, as $2x_n^2 = 2x_n \pmod{4}$. 
	Furthermore, the case $c_n=-1$ can be derived from the case $c_n=1$, since negating $h$ is equivalent to taking the complex conjugate of the root of unity and thus of the sum being calculated. We, therefore, only  
	consider the cases where $c_n \in \binset$.

	\begin{enumerate}
		\item If $c_n =0$, then we can apply Claim~\ref{claim:linearcase} to conclude that
	\begin{align*}
	\sum_{x \in \binset^n} \omega_4^{h(x)} = 2 \sum_{x' \in \binset^{n-1}} \omega_4^{g(x')} \delta_{\ell(x')}~.
	\end{align*}
If $\ell$ is a constant polynomial, then we have $2 \sum_{x' \in \binset^{n-1}} \omega_4^{g(x')} \cdot z$ for some $z\in\binset$, and the lemma follows by induction. Otherwise, we consider the binary subspace cut out by $\ell(x')=0 \pmod{2}$. Concretely, there is a binary matrix $B$, vector $d$ and variables $y = (y_1, \ldots, y_{n-2})$ such that the following subsets of $\Z_2^{n-1}$ are equal;

\[ \{ x'\in\binset^{n-1} : \ell(x')=0 \pmod{2} \} = \{ By+d \pmod{2} : y \in \binset^{n-2} \}~,\]

Therefore, in this case, 

	\[\sum_{x \in \binset^n} \omega_4^{h(x)} = 2 \sum_{y \in \binset^{n-2}} \omega_4^{g(By+d \pmod{2})} = 2 \sum_{y \in \binset^{n-2}} \omega_4^{g(By+d)}~,\]

where the second inequality follows from Claim~\ref{claim:respectmod2}. We can now apply the induction hypothesis over $g'(y)=g(By)$ (which remains respectful by Claim~\ref{claim:respectfullinear}).

\item If $c_n = 1$, we apply Lemma~\ref{lem:squarecase} to conclude that
	\begin{align*}
	\sum_{x \in \binset^n} \omega_4^{h(x)} = \sqrt{2} \zeta \sum_{x' \in \binset^{n-1}} \omega_4^{g(x')-(\ell(x'))^2}~.
\end{align*}
Applying the induction hypothesis with $g'(x') = g(x')-(\ell(x'))^2$, which is necessarily respectful (this can be viewed by change of variables or directly by opening the square), establishes the lemma in this case as well.

	\end{enumerate}
	This concludes the proof of the lemma.
\end{proof}

\end{document}